\documentclass[12pt]{article}%

\pdfoutput=1
  \usepackage[margin=1.0in]{geometry}
\usepackage{physics}
\usepackage{tkz-graph}
\usepackage{amsthm}
\usepackage{amsmath,amsfonts,graphicx,mathdots,cite,wrapfig,float}
\usepackage
{hyperref}

\usepackage{cite}
\newtheorem{theorem}{Theorem}%[section]
\newtheorem{proposition}{Proposition}%[section]

\newtheorem{corollary}[theorem]{Corollary}

\newtheorem{lemma}[theorem]{Lemma}

\DeclareMathOperator{\nn}{\mathbb{N}}
\DeclareMathOperator{\complex}{\mathbb{C}}
\usepackage{xcolor}

\newcommand{\PI}{{\ensuremath{\mbox{\small$\Pi$\hskip.012em\llap{$\Pi$}\hskip.012em\llap{$\Pi$}\hskip.012em\llap{$\Pi$}}}}}

\title{Perfect quantum state transfer on diamond fractal graphs}

\begin{document}

\author{Maxim Derevyagin$^1$, Gerald V. Dunne$^{1,2}$, Gamal Mograby$^1$\\ Alexander Teplyaev$^{1,2}$} 

\date{\today}

\maketitle

\noindent 
 {\bf Affiliation:}\\ {$^{1}$Mathematics Department, University of Connecticut, Storrs CT 06269\\
 $^{2}$Physics Department, University of Connecticut, Storrs CT 06269.\\
Email: 
{maksym.derevyagin@uconn.edu}   
{gerald.dunne@uconn.edu}  \\          
{gamal.mograby@uconn.edu}   
{alexander.teplyaev@uconn.edu}       }
 
 \begin{abstract}
 { 
 In the quest for designing novel protocols for quantum information and quantum computation, an important goal is to achieve perfect quantum state transfer for systems beyond the well-known one dimensional cases, such as 1d spin chains. We use methods from fractal analysis and probability to find a new class of quantum spin chains on fractal-like graphs (known as diamond fractals) which support  perfect quantum state transfer, and which have a wide range of different Hausdorff and spectral dimensions.
 %We extend the analysis of perfect quantum state transfer beyond one dimensional spin chains to show that it can be achieved and designed on a large class of fractal structures, known as diamond fractals,  which have a wide range of Hausdorff and spectral dimensions.
 The resulting systems are spin networks combining Dyson hierarchical model structure with transverse permutation symmetries of varying order.}
 \end{abstract}

\section{Introduction}
The study of state transfer was initiated by S.~Bose \cite{Bose03,Bose_review},  who  considered  a  $1D$  chain  of $N$
qubits coupled by the time-independent Hamiltonian. The main idea is to transport a quantum state from one end of the chain to the other. The transport
of the quantum state from one location to another is called perfect
if it is realized with probability 1, that is, without  dissipation. 
{In addition to its fundamental interest, this means that perfect quantum state transfer also has potential applications to the design of sub-protocols for quantum information and quantum computation \cite{Kay10,CVZh17,KLY17}.}
A number of one dimensional cases, 
when perfect transmission can be achieved, have been found
in some $XX$ chains with inhomogeneous couplings, see 
\cite[and references therein]{Kay10,Bose_review,christandl2004perfect,burgarth2005conclusive,burgarth2005perfect,karbach2005spin,Opperman10,Opperman12,Godsil,Godsil08,Godsil12,VZ12,qin2013high}. %In these cases, the
%probability for the transfer of a single spin excitation from one
%end of the chain to the other is found to be~1 for certain
%times. 
These models have the advantage that the perfect transfer
can be done without the need for active control. 
 Recently there has been active interest 
 to generalize these results to graphs with potentials and to graphs that are not one dimensional 
 % The progress included a number of various graphs structures
 \cite{KLY17,pretty,KMPPZ17,VZh12}. 
 {These works illustrate the fact  that perfect state transfer is a rare phenomenon, for which the construction of explicit examples remains rather non-trivial.}
% The area of the perfect state transfer has recently become very active due to its importance for the quantum information theory \cite{CVZh17,KLY17}. Although there are many papers devoted to the perfect state transfer, it is apparent that perfect state transfer is a rather rare phenomenon, and constructing examples when this occurs can be quite difficult although very important   in the field of quantum computers.  

The main result of our paper is to show that perfect quantum state transfer is possible on the large and diverse class of fractal-type diamond graphs, 
{which have different  geometrical properties including a wide range of dimensions.} 
These graphs have provided an important collection of structures with 
interesting physical and mathematical properties 
and a broad variety of geometries, see 
\cite{berker1979renormalisation,griffiths1982spin,
kaufman1984spin,
MT,ADT09,HK,BE,NT,AR18,AR19,T08,MT2,brzoska2017spectra} 
and Figures  \ref{fig:diamond}, \ref{fig:JustTheGraphnew}, \ref{fig:Diamond3GraphsNew}. 
The structure of these graphs is such that they combine spectral properties of Dyson hierarchical models and transport properties of one dimensional chains. The methods that we use are discretized versions of the methods recently developed in 
\cite{AR18,AR19} (see also \cite{AHTT18,ST18}), {which provides a construction of Green's functions for diamond fractals.}
Our work is part of a long term study of  mathematical physics on fractals and self-similar graphs 
\cite{v1,v2,ADT09,ADT10,ABDTV12,ACDRT,Akk,Dunne12,hanoi,HM20,mograby2020spectra,MDDTlax}, {in which novel features of quantum processes on fractals can be associated with the unusual spectral and geometric properties of fractals compared to regular graphs and smooth manifolds.}
% Previous figure
%\begin{figure}[h]
%	\centering
% 		\includegraphics[width=0.44\textwidth]{diamondNew.pdf}
%	\caption{The most standard  diamond hierarchical fractal graph with the similarity dimension $\text{dim}=2$,  \cite[Section 7]{v1} and \cite{ADT09,MT,HK,NT,AR18,T08}.}
%	\label{fig:diamond}
%\end{figure}

\begin{figure}[H]
	\centering
	\resizebox{13cm}{!}{\input{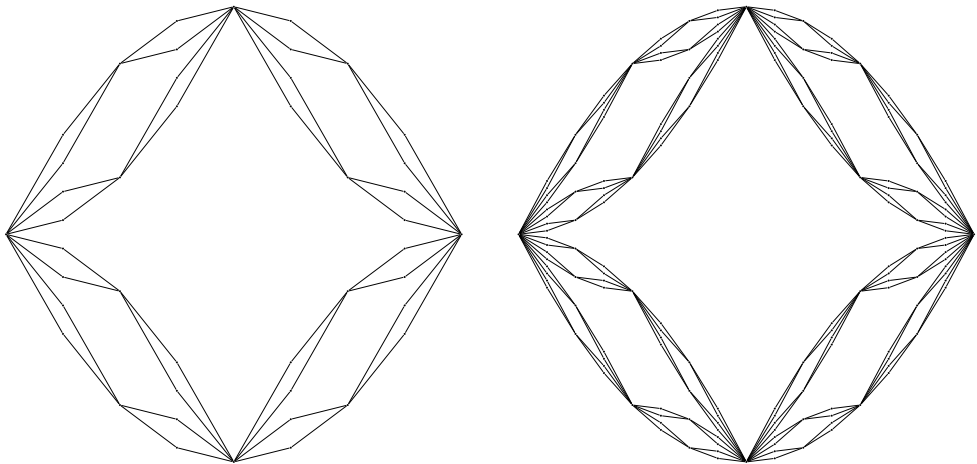}}
	\caption{The most standard  diamond hierarchical fractal graphs, levels 3 and 4, with the similarity dimension $\text{dim}=2$,  \cite[Section 7]{v1} and \cite{berker1979renormalisation,griffiths1982spin,
			kaufman1984spin,ADT09,MT,HK,BE,NT,AR18,AR19,T08}. % (Left: graph level 3, right: graph level 4)
	}
	\label{fig:diamond}
\end{figure}

\begin{figure}[H]
	\centering
	\resizebox{16cm}{!}{\input{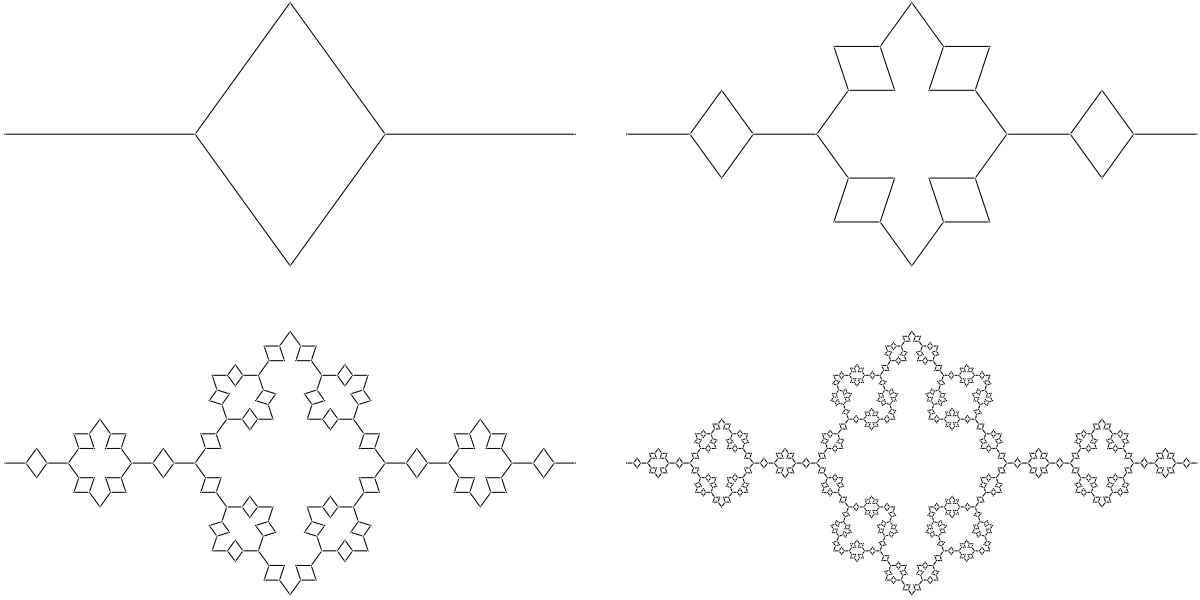}}
	\caption{Diamond graphs of level 1,2,3 and 4 with uniformly bounded degree  and the similarity dimension $\text{dim}=\frac{\log6}{\log4}$, \cite{KT,MT,MT2,LP}.}
	\label{fig:JustTheGraphnew}
\end{figure}

%Therefore a substantial part of our work   deal with 
%almost perfect state transfer, i.e.  when the state is being almost transferred or, in other words, there is an acceptable error/dissipation. 

\section{1-D Chains}
{We begin with a brief summary of perfect quantum state transfer on 1D chains \cite{Kay10}.}
%For the mathematical aspect of the   problem, we begin with the 
{Consider} one dimensional Hamiltonians ${\bf H}$ of the $XX$ type
with nearest-neighbor interactions
\begin{equation*}
\label{spinHamilton}
{\bf H}=
\tfrac{1}{2}  
\sum_{n=0}^{N-1} J_{n+1}(\sigma_n^x \sigma_{n+1}^x + \sigma_l^y
\sigma_{n+1}^y) +  
\tfrac{1}{2} %\: 
\sum_{n=0}^N B_n(\sigma_n^z +1),
\end{equation*}
where $J_n$ are the constants coupling the sites
$(n-1)$ and $n$, and $B_n$ are the strengths of  the magnetic field
at the sites $n$ ($n=0,1,\dots,N$). The symbols $\sigma_n^x, \:
\sigma_n^y,\: \sigma_n^z$ denote the standard Pauli matrices which act
as follows on the single qubit states $\ket{\downarrow}$ and $\ket{\uparrow}$:
\begin{eqnarray*}
\sigma^x \ket{\downarrow} = \ket{\uparrow},& \quad \sigma^y \ket{\downarrow} = -i \ket{\uparrow},& \quad \sigma^z \ket{\downarrow} = -\ket{\downarrow} \\
\sigma^x \ket{\uparrow} = \ket{\downarrow},& \quad \sigma^y \ket{\uparrow}  = i \ket{\downarrow},& \quad \sigma^z \ket{\uparrow} = \ket{\uparrow}.
\end{eqnarray*}
%\Gamal{Sasha, I changed the $1$ to $\uparrow$ and the zero to $\downarrow$ (spin up , spin down notation used in quantum physics), to avoid confusion with notation on next page}
It is straightforward to see that
$
[{\bf H}, %\frac{1}{2} 
\: \sum_{n=0}^N (\sigma_n^z +1)]=0
$
and so the eigenstates of ${\bf H}$ split in subspaces
labeled by  the number of spins over the chain that are in state
$\ket{\uparrow}$.  It   suffices to restrict ${\bf H}$
to the subspace spanned by the states that contain only one
excitation. A natural basis for that subspace is
given by the vectors
$
\ket{n} = (0,0,\dots, 1, \dots, 0), \quad n=0,1,2,\dots,N,
$
where the only "$\uparrow$" occupies the $n$-th position.  In this basis,
the restriction ${\bf J}$ of ${\bf H}$ to the one-excitation subspace is given
by the following $(N+1) \times (N+1)$ symmetric tridiagonal matrix
\begin{equation}
\label{spinHamiltonMatrix}
{\bf J} =
 \begin{pmatrix}
  B_{0} & J_1 &  &  & {\bf 0}\\
  J_{1} & B_{1} & J_2 & &  \\
     &  J_2 & B_2 & \ddots &    \\
   &   &  \ddots &    \ddots& J_N  \\
   {\bf 0} & &  & J_N & B_N
\end{pmatrix}.
\end{equation}
Such matrices are called Jacobi matrices and, as usual for the theory of Jacobi matrices, we assume that $J_n>0$ for $n=1$, $2$, \dots $N$. Clearly, the action of the operator ${\bf J}$ on the basis vectors $\ket{n}$ gives 
\begin{equation*}
 {\bf J} \ket{n} = J_{n+1} \ket{{n+1}} + B_n \ket{n} + J_{n}
\ket{{n-1}},
\end{equation*}
for $n=0,1,\dots, N$, where we set $J_0=J_{N+1}=0.$
Now we can see that after some time $t$ the initial state will evolve into the state
$e^{it{\bf J}}\ket{0}.$
So, in order to transfer an excitation from the site $\ket{0}$ to the site $\ket{N}$ there should exist $T>0$ and $\phi\in{\mathbb R}$ such that
\begin{equation}
\label{PSTCondition}
e^{iT{\bf J}}\ket{0}=e^{i\phi}\ket{N}.
\end{equation}

As was noted in \cite{Kay10} the latter condition immediately implies that the entries of the Jacobi matrix $J$ satisfy  
the following relations
\[
B_n = B_{N-n}, \quad J_{n} = J_{N+1-n}, \quad n= 1, 2, \dots N,
\]
which is the mirror symmetry of the matrix ${\bf J}$. This property can also be expressed in the following way
 \[
 {\bf J} ={\bf R}{\bf J}{\bf R},  
\]
where the matrix ${\bf R}$, the mirror reflection matrix, is
\[
{\bf R}=\begin{pmatrix}
  0 & 0 & \dots & 0 & 1    \\
  0 & 0 & \dots  & 1 & 0  \\
    \vdots  & \vdots & \iddots & \vdots & \vdots      \\
    0 & 1&\dots&0&0\\
   1 & 0 &  \dots & 0 &0  \\

\end{pmatrix}.
\]
Furthermore, in \cite{Kay10} the following  necessary and sufficient
conditions  for  state  transfer  in  the chain corresponding to the mirror symmetric Jacobi matrix ${\bf J}$ was proved:
the ordered set of the eigenvalues $\lambda_k$ of ${\bf J}$ ($\lambda_{k-1}<\lambda_{k}$) must
satisfy 
\begin{equation}\label{PST_cond}
\lambda_k-\lambda_{k-1}=(2m_k+1)\pi/T,\quad k=1, 2,\dots, N,
\end{equation}
where $T$ is the state transfer time, and $m_k$ is a nonnegative
integer, which can vary with $k$.

As an example we can consider one of the simplest cases of spin chains with perfect state transfer discussed in \cite{christandl2004perfect}. To this end, let us set
\begin{equation}
\label{exampleUsedatEnd}
J_n=\frac{\sqrt{n(N+1-n)}}{2}, \quad B_n=0
%, \quad n=0,1,\dots N
\end{equation}  
and  so  the underlying Jacobi matrix is mirror symmetric and it corresponds to the symmetric Krawtchouk polynomials \cite{Szego}. Also, it is known that in this case we have that
\begin{equation}
\lambda_k=k-N/2, \quad k=0,1,\dots, N, 
\end{equation}
and, thus, $\lambda_k-\lambda_{k-1}=1$, which means that the condition \eqref{PST_cond} is satisfied with $T=\pi$. As a result, the corresponding 1D spin system can realize perfect state transfer with the transfer time $T=\pi$. For more examples of spin chains with perfect transfer, see \cite{Kay10,VZh12} and references therein. 

\section{Hamiltonians on Graphs}
\label{sec:ham}
We extend the results mentioned above to a collection of fractal-type diamond graphs. {These graphs are no longer one dimensional, so they can be used to study more complex quantum systems as an extension to the 1D spin chain models.  Indeed, the diamond fractals can have a wide variety of dimensions for different choices of their self-similar structure.} We equip these graphs with a general Hamiltonian that encodes their geometric information and takes the fractal-type diamond graph symmetries into account. 
{Our main result in this paper is to show that  perfect state transfer on this collection of fractal-type diamond graphs can be reduced to an appropriately constructed 1D chain. We effectively separate variables into a longitudinal direction and transverse directions related to a hierarchy of permutation symmetries. This separation leads to
conditions that are sufficient to both {\it construct} and {\it design} these general Hamiltonians in such a way that guarantees  perfect state transfer.}

The class of fractal-type diamond graphs studied in \cite{AR18} is a family of graphs $\{G_l\}_{l \geq 0}$ which is characterized by two sequences of numbers: a sequence of {\it branching parameters} $\{ \mathcal{N}_l \}_{l \geq 0}$, and a sequence of {\it segmenting numbers} 
$\{\mathcal{J}_l \}_{l \geq 0}$. {Each link on the graph branches into a given number of links, and is also segmented into a given number of links. See Figures \ref{fig:diamond}, 
%\ref{fig:JustTheGraphnew}, 
\ref{fig:Diamond3GraphsNew} for some examples that illustrate this structure.}
These sequences generate inductively $\{G_l\}_{l \geq 0}$ in the following sense. At level $l$ we construct $G_l$ by replacing each edge from the previous level $G_{l-1}$ by $\mathcal{N}_l$ new branches, whereas each new branch is then segmented into $\mathcal{J}_l$ edges that are arranged in series. For our purposes in this paper, we will initialize $G_0$ as the one edge graph connecting two nodes. For example let $\mathcal{N}_l=\mathcal{J}_l=2$ for all levels $l \in \nn$ and $G_0$ be the one edge graph connecting a node $x_L$ with another node $x_R$. A construction of the first three levels is schematized in Fig
\ref{fig:Diamond3GraphsNew}. These graphs are fractal-type in the sense that the sequence $\{G_l\}_{l \geq 0}$  approximates a limit graph which is a special diamond fractal, see Figure \ref{fig:diamond} for higher levels \cite{ADT09}.

%\begin{figure}[h]
%	\centering
%	\includegraphics[width=0.55\textwidth]{Diamond3GraphsNew1.png}
%	\caption{$G_0$, $G_1$, $G_2$ and the mapping $\PI_2$}
%	\label{fig:Diamond3GraphsNew}
%\end{figure}

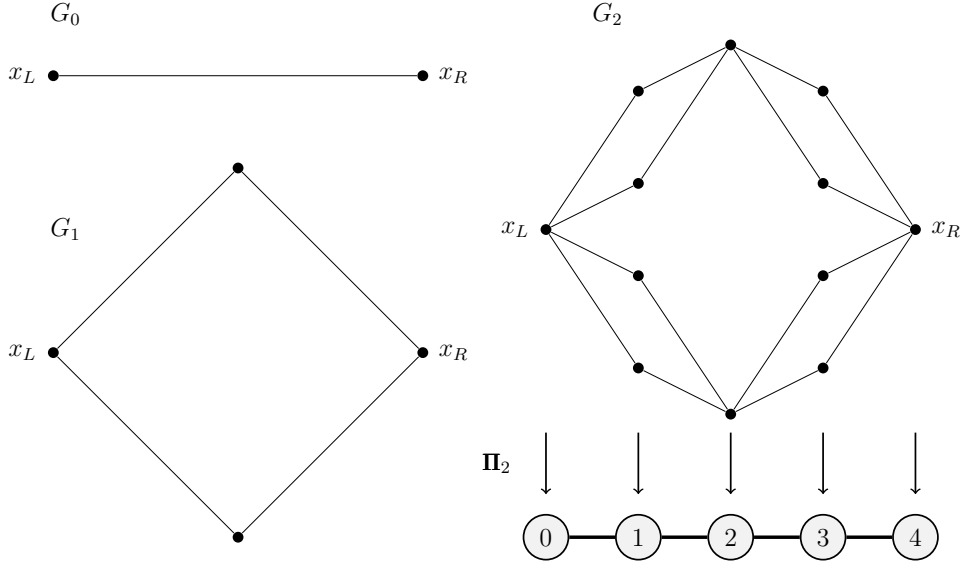
\begin{figure}[htb]
\centering
\resizebox{13cm}{!}{\begin{tikzpicture}
\tikzset{vertex/.style={circle,fill=#1,inner sep=0,minimum size=5.0pt}} 
\tikzset{edge/.style = {draw=black,   thin}}
% vertices

\node[draw=none] at (-5.5,-2){$x_L$};
\node[draw=none] at (1.5, -2){$x_R$};

% vertices
\node[vertex] (100) at  (-5.0,-2.0) {};
\node[vertex] (101) at  (1.0,-2.0) {};
\node[vertex] (102) at  (-2.0,-5.0) {};
\node[vertex] (103) at  (-2.0,1.0) {};
% edges
\draw[edge] (100) to (102);
\draw[edge] (100) to (103);
\draw[edge] (101) to (102);
\draw[edge] (101) to (103);

\node[draw=none] at (-5.5,2.5){$x_L$};
\node[draw=none] at (1.5, 2.5){$x_R$};

\node[vertex] (50) at  (-5.0,2.5) {};
\node[vertex] (51) at  (1.0,2.5) {};
\draw[edge] (50) to (51);

\node[draw=none] at (2.5,0){$x_L$};
\node[draw=none] at (9.5, 0){$x_R$};

% vertices
\node[vertex] (10) at  (3.0,0.0) {};
\node[vertex] (11) at  (9.0,0.0) {};
\node[vertex] (12) at  (6.0,-3.0) {};
\node[vertex] (13) at  (6.0,3.0) {};
\node[vertex] (14) at  (4.5,-2.25) {};
\node[vertex] (15) at  (4.5,-0.75) {};
\node[vertex] (16) at  (4.5,0.75) {};
\node[vertex] (17) at  (4.5,2.25) {};
\node[vertex] (18) at  (7.5,-2.25) {};
\node[vertex] (19) at  (7.5,-0.75) {};
\node[vertex] (20) at  (7.5,0.75) {};
\node[vertex] (21) at  (7.5,2.25) {};
% edges
\draw[edge] (10) to (14);
\draw[edge] (10) to (15);
\draw[edge] (10) to (16);
\draw[edge] (10) to (17);
\draw[edge] (11) to (18);
\draw[edge] (11) to (19);
\draw[edge] (11) to (20);
\draw[edge] (11) to (21);
\draw[edge] (12) to (14);
\draw[edge] (12) to (15);
\draw[edge] (12) to (18);
\draw[edge] (12) to (19);
\draw[edge] (13) to (16);
\draw[edge] (13) to (17);
\draw[edge] (13) to (20);
\draw[edge] (13) to (21);

\tikzset{vertex_c/.style={circle,thick,draw=black!100,fill=gray!10 , minimum size=5mm}} 
\tikzset{edge_c/.style = {draw=black, ultra thick}}
%% vertices
% vertices
\node[vertex_c] (0) at  (3.0,-5) {0};
\node[vertex_c] (1) at  (4.5,-5) {1};
\node[vertex_c] (2) at  (6.0,-5) {2};
\node[vertex_c] (3) at  (7.5,-5) {3};
\node[vertex_c] (4) at  (9.0,-5) {4};
% edges
\draw[edge_c] (0) to (1);
\draw[edge_c] (1) to (2);
\draw[edge_c] (2) to (3);
\draw[edge_c] (3) to (4);

\node[draw=none] at (-4.8,3.5){$G_0$};
\node[draw=none] at (-4.8,0){$G_1$};
\node[draw=none] at (4.0, 3.5){$G_2$};
%labels

\node[draw=none] at (2.2, -3.8){$\PI_2$};

% projection

\draw [thick,->] (9,-3.3)
        -- (9,-4.3)  ;
				
\draw [thick,->] (7.5,-3.3)
        -- (7.5,-4.3)  ;		

\draw [thick,->] (6,-3.3)
        -- (6,-4.3)  ;
				
\draw [thick,->] (4.5,-3.3)
        -- (4.5,-4.3)  ;

\draw [thick,->] (3,-3.3)
        -- (3,-4.3)  ;

\end{tikzpicture}}
	\caption{A construction of the first three levels $G_0$, $G_1$, $G_2$ and the mapping $\PI_2$}
	\label{fig:Diamond3GraphsNew}
\end{figure}

Let $V_l$ denote the set of nodes of {the diamond graph} $G_l$. We define the mapping $\PI_l:V_l \to \{0,\dots, N \}$ which assigns each node $x \in V_l$ the number of edges of the shortest path from $x$ to the node $x_L$. 
For the most standard fractal-type diamond graphs  in Figure~\ref{fig:diamond} we have $N=2^l$. This is not true in general; e.g.,  for the fractal-type graph in Figure~\ref{fig:JustTheGraphnew}  we have $N=4^l$. For level two the mapping $\PI_2$ is demonstrated in Figure \ref{fig:Diamond3GraphsNew}. 
The set of nodes $V_l$ is decomposed into 
a disjoint union of 
$N+1$ intrinsically transversal layers induced by the preimages of $\PI_l$, i.e. $V_l=\{ \PI_l^{-1}(0)\cup\PI_l^{-1}(1) \ldots \cup\PI_l^{-1}(N)\}$. In particular, when a node is in the  transversal layer $\PI_l^{-1}(n)$, then it has {an intrinsic} distance of $n$ edges to the node $x_L$.

A quantum state on $G_l$  is represented by  a complex-valued wave function on the nodes $V_l$. The space of quantum states is defined by  
\begin{equation*}
L^2(G_l)=\{\psi \ | \ \psi:V_l \to \complex \}
\end{equation*}
which is a Hilbert space equipped with the inner product
\begin{equation}
\label{key-inner}
	\bra{\psi}\ket{\varphi}_{G_l} = \sum_{x \in V_l} \psi(x) \overline{\varphi(x)} \mu_l(x)
\end{equation}
where the weights are given by $\mu_l(x)=\frac{1}{|\PI_l^{-1}(n)|}$ for $n = \PI_l(x)$ and $|\PI_l^{-1}(n)|$ denotes the number of nodes in the transversal layer $\PI_l^{-1}(n)$ that contains $x$. {This factor accounts for the transverse degeneracies due to the permutation symmetries at a given level $l$.} We  denote by $\ket{x}$ the wave function that assigns one to the node $x \in V_l$ and zeros elsewhere (one-excitation state on $G_l$). The set of all one-excitation states $\{ \ket{x_L}, \ldots, \ket{x_R} \}$ form a natural basis for $L^2(G_l)$.

An $l$-level Hamiltonian on $G_l$ is a Hermitian operator $\mathbf{H}_l$ acting on $L^2(G_l)$. To encode the geometric information of the fractal-type diamond graph $G_l$ in the Hamiltonian, we impose the following assumptions on $\mathbf{H}_l$:
\begin{itemize}
\item {{\it Nearest-neighbor coupling}:} for $x,y \in V_l$, let $\bra{x}\mathbf{H}_l\ket{y}_{G_l}=0$
if $x$ and $y$ are not connected by an edge, i.e.  the transition matrix element from the quantum state $\ket{y}$ to $\ket{x}$ is zero if the nodes $y$ and $x$ are not adjacent in $G_l$.
\item 
{{\it Symmetric coupling}:} for $x_1,y_1,x_2,y_2 \in V_l$ such that both $x_1,y_1$ and $x_2,y_2 $ are adjacent, let
\begin{align*}
\bra{x_1}\mathbf{H}_l\ket{y_1}_{G_l}=\bra{x_2}\mathbf{H}_l\ket{y_2}_{G_l} \\  \text{if } \PI_l(x_1)=\PI_l(x_2) \text{ and } \PI_l(y_1)=\PI_l(y_2),
\end{align*}
i.e.  the transition matrix elements are compatible with the intrinsically transversal layers of $G_l$.
\end{itemize}
{This means that we can} regard $\{0,\dots, N \}$ as the set of nodes of a 1D chain.
To reduce the perfect state transfer problem from the graph $G_l$ to this 1D chain, we introduce the following Hilbert space $L^2(\{0,\dots, N \})=\{\psi \ | \ \psi:\{0,\dots, N \} \to \complex \}$
equipped with the inner product
\begin{equation}
	\label{standard inner}
	\bra{\psi}\ket{\varphi}_l = \sum_{n=0}^N \psi(n) \overline{\varphi(n)}.
\end{equation}
Moreover we project a wave function in $L^2(G_l)$ to a wave function in $L^2(\{0,\dots, N \})$ through averaging its values on the transversal layers, 
\begin{eqnarray*}
P_l: L^2(G_l) & \to & L^2(\{0,\dots, N \}) \\
\psi &\mapsto &  P_l\psi(n)=\frac{1}{|\PI_l^{-1}(n)|}\sum_{x \in \PI_l^{-1}(n)}  \psi(x).
\end{eqnarray*}
A  simple calculation using the definition of the inner products gives $\bra{P_l\psi}\ket{\varphi}_l = \bra{\psi}\ket{P^\ast_l \varphi}_{G_l}$, where the adjoint operator $P^{\ast}_l $ of $P_l$ is defined by
\begin{eqnarray*}
P^{\ast}_l: L^2(\{0,\dots, N \})  \to  L^2(G_l) \\
\varphi  \mapsto P^{\ast}_l \varphi (x) = \varphi(\PI_l(x)).
\end{eqnarray*}
\begin{figure}[h]
	\centering
		\includegraphics[width=0.5\textwidth]{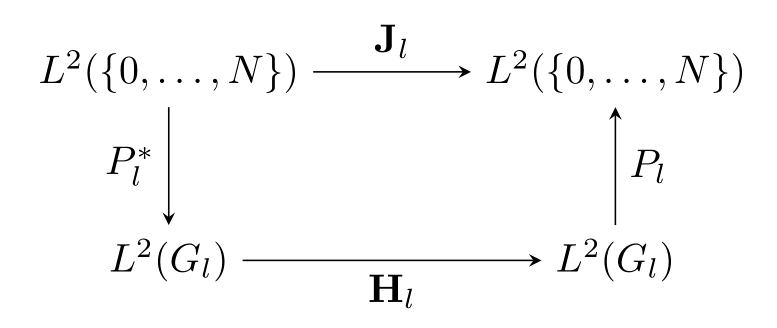}
	\caption{A diagram to explain the mapping between the Hamiltonian $\mathbf{H}_l$ on the diamond fractal graph $G_l$ and the effective Hamiltonian $\mathbf{J}_l$ on the associated one-dimensional spin chain. The definition of $\mathbf{J}_l=P_l \mathbf{H}_l P_l^{\ast}$: we apply first $P_l^{\ast}$, then $\mathbf{H}_l$, and finally $P_l$.}
	\label{fig:commutativeDiagram}
\end{figure}

\section{Main results and proofs}

The Hamiltonian $\mathbf{H}_l$ on $G_l$ induces an operator on the 1D chain $\{0,\dots, N \}$ by
$$\mathbf{J}_l  =  P_l \mathbf{H}_l P_l^{\ast}$$ which acts on $L^2(\{0,\dots, N \})$, see Figure \ref{fig:commutativeDiagram}. We denote the one-excitation states in $L^2(\{0,\dots, N \})$ by $\ket{n}$ for a node $n \in \{0,\dots, N \}$. Let $\mathbf{H}_l = (H_l(x,y))_{x,y \in V_l}$ be the matrix representation of $\mathbf{H}_l$ with respect to $\{ \ket{x_L}, \ldots, \ket{x_R} \}$. The following result relates the matrix elements of $\mathbf{H}_l$ to $\mathbf{J}_l$.  
\begin{proposition} \label{result1}
Let $x \in V_l$,
\begin{enumerate}
	\item $H_l(x,x) = \bra{\PI_l(x)}\mathbf{J}_l \ket{\PI_l(x)}_l$.
	\item Let $y \in V_l$ be adjacent to $x$ and $\PI_l(y)=\PI_l(x)\pm 1$, then
$$H_l(x,y) = \frac{1}{\mathbf{deg}_{\pm}(x) }\bra{\PI_l(x)}\mathbf{J}_l \ket{\PI_l(x)\pm 1}_l,$$ 
where $\mathbf{deg}_{\pm}(x)$ is defined as follows: let $x \in \PI_l^{-1}(n)$ for some $n \in \{0,\ldots, N-1 \}$ the mapping $\mathbf{deg}_{+}(x)$  assigns the node $x$ the number of edges that connect $x$ to nodes in  $\PI_l^{-1}(n+1)$. Similarly $\mathbf{deg}_{-}(x)$ assigns the node $x$ the number of edges that connect $x$ to nodes in $\PI_l^{-1}(n-1)$.
\item The Hamiltonian $\mathbf{H}_l$ is self-adjoint with respect to the inner product \ref{key-inner} if and only if $\mathbf{J}_l$ is self-adjoint with respect to the inner product \ref{standard inner}. 
\end{enumerate}
\end{proposition}
The following result justifies the reduction of the perfect transfer problem from $G_l$ to a 1D chain.
\begin{theorem} \label{result2}
If the   perfect state transfer on the 1D chain $\{0,\dots, N \}$ is achieved, i.e.  there exists $T_l>0$ such that $$e^{iT_l \mathbf{J}_l} \ket{0}=e^{i\phi}\ket{N}$$ for some phase $\phi$, then the  perfect state transfer on $G_l$ is also achieved with the same time $T_l$ and phase $\phi$, i.e.  
\begin{equation*}
e^{i T_l \mathbf{H}_l } \ket{x_L}= e^{i\phi}\ket{x_R}
\text{ \ and \ }
e^{i T_l \mathbf{H}_l } \ket{x_R}= e^{i\phi}\ket{x_L}.
\end{equation*}
\end{theorem}

%\section{Proofs}

 For the purpose of proving the main results, we introduce first the following auxiliary definitions and lemmas. We define the space of functions $\psi \in L^2(G_l)$ that are constant on each transversal layer $\PI_l^{-1}(n)$ for $n \in \{0,\ldots, N \}$ and denote it by
 \begin{equation*}
 L^2_{sym}(G_l)=\{\psi \ | \  \psi(x)=\psi(y) \text{ if } \PI_l(x)=\PI_l(y) \}.
 \end{equation*}
 $L^2_{sym}(G_l)$ is a subspace of $L^2(G_l)$ and let $\mathbf{Proj}_l: L^2(G_l) \to L^2_{sym}(G_l)$ be the projection of $L^2(G_l)$ onto $L^2_{sym}(G_l)$. 
 \begin{lemma}
 	\label{invariantSubspaceLemma}
 	$L^2_{sym}(G_l)$ is an invariant subspace of $L^2(G_l)$ under $\mathbf{H}_l$.
 \end{lemma}
 \begin{proof}
Let $ \PI_l^{-1}(n)=\{x_{n_1}, \ldots, x_{n_m} \}$ for some $n \in \{0,\ldots, N \}$ and $\ket{\psi}=\ket{x_{n_1}}+\ldots+\ket{x_{n_m}}$, i.e., for $x \in V_l$
\begin{equation*}
 \bra{x}\ket{\psi}_{G_l}=
  \begin{cases}
    \frac{1}{|\PI_l^{-1}(n)|}      & \quad \text{if } x \in \PI_l^{-1}(n) \\
   \quad  0 & \quad \text{otherwise } 
  \end{cases}
\end{equation*}
It suffices to show that $\mathbf{H}_l\ket{\psi} \in L^2_{sym}(G_l)$. By the symmetric coupling assumption on $\mathbf{H}_l$ we set $c_n=\bra{x_{n_1}}\mathbf{H}_l\ket{x_{n_1}}_{G_l}= \ldots = \bra{x_{n_m}}\mathbf{H}_l\ket{x_{n_m}}_{G_l}$. Note that any two nodes in the same transversal layer are not adjacent. Similarly, for the neighboring transversal layers, we set $c_{n-1}=\bra{x}\mathbf{H}_l\ket{y}_{G_l}$ for any adjacent nodes $x \in \PI_l^{-1}(n-1),y \in \PI_l^{-1}(n)$ for $n>0$ and $c_{n+1}=\bra{x}\mathbf{H}_l\ket{y}_{G_l}$ for any adjacent nodes $x \in \PI_l^{-1}(n+1), y \in \PI_l^{-1}(n)$ for $n < N$. One can easily verify the following formula,
\begin{equation*}
 \bra{x}\mathbf{H}_l\ket{\psi}=
  \begin{cases}
		\mathbf{deg}_{+, n-1}c_{n-1}      & \quad \text{if } x \in \PI_l^{-1}(n-1) \\
    c_n      & \quad \text{if } x \in \PI_l^{-1}(n) \\
		\mathbf{deg}_{-, n+1}c_{n+1}      & \quad \text{if } x \in \PI_l^{-1}(n+1) \\
    0 & \quad \text{otherwise } 
  \end{cases}
\end{equation*}
where the last case is implied by the nearest-neighbor coupling assumption on $\mathbf{H}_l$ and $\mathbf{deg}_{+, n-1}$ (or $\mathbf{deg}_{-, n+1}$) is the number of edges that connect a node in  $\PI_l^{-1}(n-1)$ (or $\PI_l^{-1}(n+1)$) to nodes in  $\PI_l^{-1}(n)$.
 \end{proof}
 \begin{lemma}
 	\label{rangeAdjointlemma}
 	The range of $P^{\ast}_l$ is $L^2_{sym}(G_l)$.
 \end{lemma}
 \begin{proof}
 	It follows by the definition of $P^{\ast}_l$ and $L^2_{sym}(G_l)$. 
 \end{proof}
 \begin{corollary}
 	\label{orthoComplementSym}
 	$ Ker P_l = (L^2_{sym}(G_l))^{\bot}$. In particular, if $\psi \in (L^2_{sym}(G_l))^{\bot}$, then the sum over a  transversal layer gives $\sum_{x \in \PI_l^{-1}(n)}  \psi(x)=0$ for $ n \in \{0,\dots, N \}$.
 \end{corollary}
 \begin{proof}
 	It follows with $ Ker P_l = (Range \ P^{\ast}_l)^{\bot}$ and Lemma \ref{rangeAdjointlemma}.
 \end{proof}
 \begin{lemma}
 	\label{PstarPlemma}
 	Let $\psi \in L^2(G_l)$ and $x \in V_l$, then
 	$P_l^{\ast} P_l \psi(x)=  \mathbf{Proj}_l \psi(x)$.
 \end{lemma}
 \begin{proof}
 	We decompose $\psi=\psi_{sym} + \psi_{sym}^{\bot}$ such that $\psi_{sym} \in L^2_{sym}(G_l)$ and $\psi_{sym}^{\bot} \in (L^2_{sym}(G_l))^{\bot}$. Corollary \ref{orthoComplementSym} implies $P_l^{\ast} P_l \psi = P_l^{\ast} P_l \psi_{sym}$. Let $n \in \{0,\dots, N \}$ and $x\in \PI_l^{-1}(n)$. Then $\psi_{sym}$ is constant on the transversal layer $\PI_l^{-1}(n)$ and it's averaging gives $P_l \psi_{sym}(n) =\psi_{sym}(x)$.
Hence, by the definition of $P_l^{\ast}$, it follows  $P_l^{\ast} P_l \psi_{sym}(x) =  \psi_{sym}(x)=\mathbf{Proj}_l \psi(x)$.
 \end{proof}
\begin{proof}[Proof of Proposition \ref{result1}]
Let $ x \in \PI_l^{-1}(n)=\{x_{n_1}, \ldots, x_{n_m} \}$ for some $n \in \{0,\ldots, N \}$. We evaluate the matrix element,
\begin{equation*}
\bra{\PI_l(x)}\mathbf{J}_l \ket{\PI_l(x)}_l =  \bra{\PI_l(x)} P_l \mathbf{H}_l P_l^{\ast} \ket{\PI_l(x)}_l = \bra{P_l^{\ast} \PI_l(x)}  \mathbf{H}_l  \ket{P_l^{\ast} \PI_l(x)}_{G_l}  
\end{equation*}
where we adopt the notation $\ket{P_l^{\ast} \PI_l(x)} = P_l^{\ast} \ket{\PI_l(x)} =\ket{x_{n_1}}+\ldots+\ket{x_{n_m}}$. The assumptions on $\mathbf{H}_l$ imply
\begin{equation*}
\bra{\PI_l(x)}\mathbf{J}_l \ket{\PI_l(x)}_l =   \bra{x_{n_1}}  \mathbf{H}_l  \ket{x_{n_1}}_{G_l} +\ldots+ \bra{x_{n_m}}  \mathbf{H}_l  \ket{x_{n_m}}_{G_l} = |\PI_l^{-1}(n)| \bra{x}  \mathbf{H}_l  \ket{x}_{G_l} = H_l(x,x)
\end{equation*}
where the last equality holds as $|\PI_l^{-1}(n)|$ cancels the weights in the inner product defined on $G_l$. Similar reasoning will give the second part of the statement. Assume $y$ is adjacent to $x$ such that $ y \in \PI_l^{-1}(n+1)=\{y_{n_1}, \ldots, y_{n_k} \}$,
\begin{eqnarray}
\bra{\PI_l(x)}\mathbf{J}_l \ket{\PI_l(x)+1}_l &= &  \bra{P_l^{\ast} \PI_l(x)}  \mathbf{H}_l  \ket{P_l^{\ast}( \PI_l(x)+1)}_{G_l} \nonumber \\
&=& (\bra{x_{n_1}}+\ldots+\bra{x_{n_m}} )\  \mathbf{H}_l \ ( \ket{y_{n_1}}+\ldots+\ket{y_{n_k}}) \nonumber \\
&=& \mathbf{deg}_{+}(x) \  |\PI_l^{-1}(n)|  \ \bra{x}  \mathbf{H}_l  \ket{y}_{G_l} \nonumber \\
&=& \mathbf{deg}_{+}(x) \ H_l(x,y) \nonumber
\end{eqnarray}
To prove the third statement, we first observe that the previous computations verify the following equation,
\begin{equation*}
\bra{\PI_l(x)}\mathbf{J}_l \ket{\PI_l(x)+1}_l 
= \mathbf{deg}_{+}(x) \  |\PI_l^{-1}(n)|  \ \bra{x}  \mathbf{H}_l  \ket{y}_{G_l}
\end{equation*}
Similarly we can show,
\begin{equation*}
\bra{\PI_l(x)+1}\mathbf{J}_l \ket{\PI_l(x)}_l 
= \mathbf{deg}_{-}(y) \  |\PI_l^{-1}(n+1)|  \ \bra{y}  \mathbf{H}_l  \ket{x}_{G_l}.
\end{equation*}
Hence, it suffices to prove  $\mathbf{deg}_{-}(y) \  |\PI_l^{-1}(n+1)| = \mathbf{deg}_{+}(x) \  |\PI_l^{-1}(n)| $. The last equality holds as both left-hand side, and the right-hand side give the number of edges between the transversal layers $\PI_l^{-1}(n)$ and $\PI_l^{-1}(n+1)$.
\end{proof}
 \begin{proof}[Proof of Theorem \ref{result2}]
 	We observe 
 	\begin{eqnarray*}
 		(\mathbf{J}_l)^k  =  P_l \mathbf{H}_l P_l^{\ast} \cdot P_l \mathbf{H}_l P_l^{\ast} \cdots P_l \mathbf{H}_l P_l^{\ast} 
 		 =  P_l (\mathbf{H}_l \mathbf{Proj}_l)^k P_l^{\ast}
 	\end{eqnarray*}
 	where $P_l^{\ast} =\mathbf{Proj}_l P_l^{\ast}$ holds due to Lemma \ref{rangeAdjointlemma} and $P_l^{\ast} P_l =  \mathbf{Proj}_l$ due to Lemma \ref{PstarPlemma}. Hence,
 	\begin{eqnarray*}
 		P_l e^{iT_l \mathbf{H}_l \mathbf{Proj}_l} P_l^{\ast} \ket{0}=  P_l e^{i\phi} P_l^{\ast}\ket{N} \quad \Rightarrow \quad P_l e^{iT_l \mathbf{H}_l \mathbf{Proj}_l} \ket{x_L}=  P_l e^{i\phi}\ket{x_R}
 	\end{eqnarray*}
which implies $e^{iT_l \mathbf{H}_l \mathbf{Proj}_l} \ket{x_L}- e^{i\phi}\ket{x_R} \in Ker(P_l)$. Note that $\ket{x_L}=P_l^{\ast} \ket{0}$ and $\ket{x_R}=P_l^{\ast}\ket{N}$.
 	By Corollary \ref{orthoComplementSym} and by the fact that $\ket{x_L}, \ket{x_R} \in L^2_{sym}(G_l)$ we conclude, 
 	\begin{equation*}
 	e^{iT_l \mathbf{H}_l \mathbf{Proj}_l} \ket{x_L}= e^{i\phi}\ket{x_R}
 	\end{equation*} 
 	Let $\mathbf{Proj}_l^{\bot}$ be the projection of $L^2(G_l)$ onto $(L^2_{sym}(G_l))^{\bot}$. We observe, 
 	$$(\mathbf{H}_l \mathbf{Proj}_l + \mathbf{H}_l \mathbf{Proj}_l^{\bot})  \ket{x_L} = \mathbf{H}_l \mathbf{Proj}_l \ket{x_L},$$ which implies $e^{iT_l \mathbf{H}_l } \ket{x_L}= e^{i\phi}\ket{x_R}$.
 	%\begin{equation*}
 	
 	%\end{equation*}
 \end{proof}

\section{Conclusions and outlook.}  
%	Using the values in \eqref{exampleUsedatEnd} for the simplest case of spin chain and the Hamiltonian construction in Proposition \ref{result1} and  \ref{result2} give    examples, see Figures \ref{fig:diamond}
%	and  \ref{fig:JustTheGraphnew}, of   quantum systems on graphs $G_l$  in which   perfect state transfer is achieved.
% Our approach is   general and   applicable to many fractal-type graphs,
%	which allows to construct   infinitely many examples. 

{Our construction provides an infinite set of new examples of novel geometries for which perfect quantum state transfer can be achieved. 
For example, using the spin-coupling values $J_n$  in \eqref{exampleUsedatEnd} for the simplest case of a spin chain, combined with our Hamiltonian construction in Proposition \ref{result1} and Theorem \ref{result2}, we find perfect quantum state transfer on diamond fractals such as those shown in Figures \ref{fig:diamond} and  \ref{fig:JustTheGraphnew}.
This clearly generalizes to the set of quantum systems on the graphs $G_l$ described in Section \ref{sec:ham}.
%  in which   perfect state transfer is achieved.
The basic projection idea  is  very general and   applicable to many other fractal-type graphs,
which allows to construct  further examples.  This opens up the possibility to {\it design} perfect quantum state transfer  on fractal-like structures with special features. The existence of results for Green's functions for these structures means that other quantum information properties such as fidelity and entanglement can be studied for these fractal structures. } 	
Moreover, our approach allows to consider other transport phenomena involving linear and nonlinear, classical and quantum waves on certain graphs, quantum graphs, and fractals. This will be the subject of future research
\cite{mograby2020spectra}.

%\begin{figure}[th!]
%	\centering
%	\includegraphics[width=0.55\textwidth]{justthegraphNew.png}
%	\caption{A diamond graph with uniformly bounded degree  and the similarity dimension $\text{dim}=\frac{\log6}{\log4}$, \cite{KT,MT,MT2,LP}.}
%	\label{fig:JustTheGraphnew}
%\end{figure}

%\end{conclusion}

\section{Acknowledgments}
This research was supported in part 
by the University of Connecticut 
Research Excellence Program, by 
DOE grant DE-SC0010339  
and by NSF DMS grants 1613025 and 2008844.  
The authors are grateful to 
Eric Akkermans,   
Patricia Alonso-Ruiz 
and 
Gabor  Lippner 
for interesting and helpful discussions. 
 The authors are grateful  to   anonymous referees for  suggested improvements to the paper.
 
 \bibliographystyle{unsrt}
 \bibliography{Quantum-State-Refs}

\end{document}